\documentclass[11pt]{article}
\usepackage{latexsym}
\usepackage{tabularx,booktabs,multirow,delarray,array}
\usepackage{graphicx,amssymb,amsmath,amssymb,mathrsfs}

\newenvironment{proof}{\noindent {\textbf{Proof:}}\rm}{\hfill $\Box$
\rm}

\def\calA{\mathcal{A}}

\newtheorem{Theo}{Theorem}
\newtheorem{Lem}{Lemma}

\begin{document}

\baselineskip=14.0pt

\title{
\vspace*{-0.55in}
%Improved Algorithms for Barrier Coverage with Mobile Sensors
Optimal Point Movement for Covering Circular Regions
}
%\medskip
%\author{Danny Z. Chen$^{1}$, Xuehou Tan$^{2}$, Haitao Wang$^{1}$ and Gangshan Wu$^{3}$ \\
%\\
%$^{1}$ University of Notre Dame, Notre Dame, IN 46556, USA \\
%$^{2}$ Tokai University, 4-1-1 Kitakaname, Hiratsuka 259-1292, Japan \\
%$^{3}$ State Key Lab. for Novel Software Technology, Nanjing University, Nanjing, China}
%dchen@cse.nd.edu, tan@wing.ncc.u-tokai.ac.jp, hwang6@nd.edu}

\author{
Danny Z. Chen\thanks{Dept.~of Computer Science and Engineering,
University of Notre Dame, Notre Dame, IN 46556, USA; { \{dchen,
hwang6\}@nd.edu.} These authors' work was supported in part by NSF
under Grant CCF-0916606. } \hspace*{0.1in} Xuehou Tan\thanks{Tokai
University, 4-1-1 Kitakaname, Hiratsuka 259-1292, Japan; {
tan@wing.ncc.u-tokai.ac.jp}.} \hspace*{0.1in} Haitao
Wang\footnotemark[1]  \thanks{Corresponding author.} \hspace*{0.1in}
Gangshan Wu\thanks{State Key Lab.~for Novel Software Technology,
Nanjing University, Nanjing 210093, China; {
gswu@graphics.nju.edu.cn}.}}

\date{ }
\maketitle

\pagestyle{plain}
\pagenumbering{arabic}
\setcounter{page}{1}

\vspace*{-0.3in}
\begin{abstract}
Given $n$ points in a circular region $C$ in the plane, we study the problems of
moving the $n$ points to its boundary to form a regular $n$-gon such that
the maximum (min-max) or the sum (min-sum) of the Euclidean
distances traveled by the points is minimized. The problems have
applications, e.g., in mobile sensor barrier coverage of wireless
sensor networks.
%Monitoring and surveillance are important aspects in modern wireless
%sensor networks, where
%sensors often need to move quickly from the interior of a
%region $C$ to the boundary of $C$, so as to form a
%barrier coverage of $C$.
%In this paper, we study the problem of moving $n$ sensors in a
%circular region to its boundary to form a regular $n$-gon such that
%the maximum (min-max) or the sum (min-sum) of the Euclidean
%distances traveled by the sensors is minimized.
The min-max problem further has two versions: the decision version and optimization
version. For the min-max problem, we present an $O(n\log^2 n)$ time
algorithm for the decision version and an $O(n\log^3 n)$ time algorithm for
the optimization version. The
previously best algorithms for the two problem versions take $O(n^{3.5})$ time
and $O(n^{3.5}\log n)$ time, respectively. For the min-sum problem,
we show that a special case with all points initially lying on the
boundary of the circular region can be solved in $O(n^2)$ time,
improving a previous $O(n^4)$ time solution. For the general
min-sum problem, we present a $3$-approximation $O(n^2)$ time
algorithm, improving the previous $(1+\pi)$-approximation $O(n^2)$ time algorithm.
A by-product of our techniques is an algorithm for dynamically
maintaining the maximum matching of a circular convex bipartite
graph; our algorithm can handle each vertex insertion or deletion on
the graph in $O(\log^2 n)$ time. This result is interesting in its
own right.
%{\it Keywords}: Computational geometry; Barrier coverage;
%Min-max problem; Min-sum problem; Maximum matching
\end{abstract}

\section{Introduction}

Given $n$ points in a circular region $C$ in the plane, we study the problems of
moving the $n$ points to its boundary to form a regular $n$-gon such that
the maximum (min-max) or the sum (min-sum) of the Euclidean
distances traveled by the points is minimized. The problems have
applications, e.g., in mobile sensor barrier coverage of wireless
sensor networks. The problems have
been studied before. In this paper we present new algorithms that significantly
improve the previous solutions for the problems.

\subsection{Problem Definitions}

Let $|ab|$ denote the Euclidean length of the line segment with two
endpoints $a$ and $b$ in the plane. Let $C$ be a circular region in
the plane. Given a set of $n$ points $S=\{A_0,A_1,\ldots,A_{n-1}\}$ in
$C$ (i.e., in its interior or on its boundary), we wish to move all sensors to
$n$ points $A'_0,A'_1,\ldots, A'_{n-1}$ on the boundary of $C$ that
form a regular $n$-gon. The {\em min-max} problem aims to minimize
the maximum Euclidean distance traveled by all points, i.e.,
$\max_{0\leq i\leq n-1}\{|A_{i} A'_{i}|\}$. The {\em min-sum}
problem aims to minimize the sum of the Euclidean distances traveled
by all points, i.e., $\sum_{i=0}^{n-1} |A_{i} A'_{i}|$.

%Since each sensor has only limited energy, it may not be able
%to move as far a distance as we want.
%Thus, we have the following {\em feasibility problem}:
Further, given a value $\lambda\geq 0$, the {\em decision version} of
the min-max problem is to determine whether it is possible to move all points in
$S$ to the boundary of $C$ to form a regular $n$-gon such that the distance
traveled by each point is no more than $\lambda$. Indeed,
let $\lambda_C$ be the maximum distance traveled by the points in
an optimal solution for the min-max problem. Then, the answer to the
feasibility problem is ``yes" if and only if $\lambda_C\leq
\lambda$. For discrimination, we refer to the original min-max
problem as the {\em optimization version} of the min-max problem.

For the min-sum problem, if the points in $S$ are given initially all lying on the
boundary of $C$, then this case is referred to as the {\em boundary
case} of the min-sum problem.

\subsection{Applications in Wireless Sensor Networks}

A Wireless Sensor Network (WSN) is composed of a large number of
sensors which monitor some surrounding environmental phenomenon.
Usually, the sensors are densely deployed either inside the target phenomenon or
are very close to it \cite{ref:AkyildizWi02}. Each sensor
is equipped with a sensing device
%, a processor of low computational capacity,
%a short-range wireless transmitter-receiver, and a
with limited battery-supplied energy.
The sensors process data obtained and forward the data to a base station.
%These characteristics of WSN require that their sensor network protocols and
%algorithms possess self-organizing capabilities, i.e.,
%sensors are able to cooperate with each other in
%order to organize and perform networking tasks together efficiently.
 A typical type of WSN applications
 %is to monitor a specified region either
 %for measuring purposes or for reporting various kinds of activities
 %(e.g., fire alarms, calamities, etc). Another type of applications
 is concerned with security and safety systems, such as detecting intruders
(or movement thereof) around infrastructure facilities and regions. Particularly,
it is often used to monitor a protected area so as to detect intruders
as they penetrate the area or as they cross the area border. For example, research
efforts have been under way to extend the scalability of wireless sensor networks
to the monitoring of international borders \cite{ref:Hu08,ref:KumarBa07}.

The study of {\it barrier coverage} using mobile sensors was originated
in \cite{ref:ChenDe07,ref:KumarBa07} and later in \cite{ref:BhattacharyaOp09}.
Different from the traditional concept of {\it full coverage},
it seeks to cover the deployment region by guaranteeing that there is no
path through the region that can be traversed undetectedly by an intruder,
i.e., all possible crossing paths through the region
are covered by the sensors \cite{ref:BhattacharyaOp09,ref:ChenDe07,ref:KumarBa07}.
%Since mobile sensors are allowed to move inside the deployment region,
%a crossing path may occur occasionally.
Hence, an interesting problem is
to reposition the sensors quickly so as to repair the existing security hole and
thereby detect intruders \cite{ref:BhattacharyaOp09}. Since barrier coverage
requires fewer sensors for detecting intruders, it gives a good
approximation of full area coverage. The planar region on which the sensors move
is sometimes represented by a circle.
Since sensors have limited battery-supplied energy, we wish
to minimize their movement. Thus, if each sensor is
represented as a point, the problem is exactly our optimal point
movement min-max (the optimization version) or min-sum problem.
Further, if each sensor has energy $\lambda$ and we want to determine
whether this level of energy is sufficient to form a barrier coverage,
then the problem becomes the decision version of the min-max problem.

%Since each sensor has only limited energy, it may not be able
%to move as far a distance as we want.
%Thus, we have the following {\em feasibility problem}:

%Note that we assume that the range of each sensor's transmitter-receiver
%is always longer than an edge of a regular $n$-gon of $C$; otherwise,
%a barrier coverage of $C$ is impossible.

\subsection{Previous Work and Our Results}

For the min-max problem, Bhattacharya {\em et al.}~\cite{ref:BhattacharyaOp09}
proposed an $O(n^{3.5})$ time algorithm
for the decision version and an $O(n^{3.5} \log n)$ time algorithm
for the optimization version, where the decision algorithm is
based on some observations and brute force and the optimization
algorithm is based on parametric search approach
\cite{ref:ColeSl87,ref:MegiddoAp83}.
Recently, it was claimed in \cite{ref:TanNe10} that these two
problem versions were solvable in $O(n^{2.5})$ time and
$O(n^{2.5}\log n)$ time, respectively. However, it seems that the
announced algorithms in \cite{ref:TanNe10} contain errors (which
might be fixed, say, by using the methods given in this paper).
In this paper, we solve the decision version in $O(n\log^2 n)$ time and the
optimization version in $O(n\log^3n)$ time, which significantly
improve the previous results. The improvements of our algorithms are
based on new observations and interesting techniques.
%It should be noted
%that our main ideas are different from those in \cite{ref:TanNe10}.
%Further, our algorithms do not rely on parametric search and are
%much simpler and easier to implement than those in
%\cite{ref:BhattacharyaOp09}.

A by-product of our techniques that is interesting in its own right
is an algorithm for dynamically maintaining the maximum matchings of
{\em circular convex bipartite graphs}. Our
algorithm handles each (online) vertex insertion or deletion on
an $n$-vertex circular convex bipartite graph in $O(\log^2 n)$ time.
This matches the performance of the best known
dynamic matching algorithm for {\em convex bipartite graphs}
\cite{ref:BrodalDy07}. Note that convex bipartite graphs are a subclass of
circular convex bipartite graphs \cite{ref:LiangCi95}.  To our best knowledge,
no dynamic matching algorithm for circular convex bipartite graphs was known before.
Since dynamically maintaining the maximum
matching of a graph is a basic problem, our result
may find other applications.
%Also, we generalize the techniques in
%\cite{ref:ColeAn89} to compute the $k$-th highest vertex in the
%arrangement \cite{ref:deBergCo08} of $O(n)$ special curves for solving the
%min-max optimization problem.

For the min-sum problem, an $O(n^2)$ time approximation algorithm with approximation
ratio $1+\pi$ was given in \cite{ref:BhattacharyaOp09}. A PTAS approximation
algorithm, which has a substantially larger polynomial time bound,
was also given in \cite{ref:BhattacharyaOp09}. In this paper,
we present an $O(n^2)$ time approximation algorithm with approximation
ratio $3$, which improves the $(1+\pi)$-approximation result in
\cite{ref:BhattacharyaOp09}. However, whether the general min-sum problem is
NP-hard is still left open.

For the boundary case of the min-sum problem, an $O(n^4)$ time
(exact) algorithm was given in \cite{ref:TanNe10}. We show that the
time bound of that algorithm can be reduced to $O(n^2)$.

The rest of this paper is organized as follows. Our algorithm for the
decision version of the min-max problem is given in Section
\ref{sec:decision}, and our algorithm for the optimization version
is presented in Section \ref{sec:optimization}. The min-sum problem
is discussed in Section \ref{sec:minsum}.

To distinguish from a normal point in the plane, in the following
paper we refer to each
point $A_i\in S$ as a sensor.

\section{The Decision Version of the Min-max Problem}
\label{sec:decision}

For simplicity, we assume the radius of the circle $C$ is $1$. Denote by
$\partial C$ the boundary of $C$. Let $\lambda_C$ be the
maximum distance traveled by the sensors in $S$ in an optimal solution for
the min-max problem, i.e., $\lambda_C =\min \{\max_{0\leq i\leq n-1}
\{|A_{i} A'_{i}| \}\}$. Since the sensors are all in $C$, $\lambda_C \leq 2$.
In this section, we consider the decision version of the min-max
problem on $C$: Given a value $\lambda$, determine whether $\lambda_C\leq
\lambda$. We present an $O(n\log^2 n)$ time algorithm for this problem.

\subsection{An Algorithm Overview}

We first discuss some concepts. A
{\it bipartite} graph $G = (V_1, V_2, E)$ with $|V_1|=O(n)$ and $|V_2|=O(n)$ is
{\it convex} on the vertex set $V_2$ if there is a linear ordering
on $V_2$, say, $V_2=\{v_0,v_1,\ldots,v_{n-1}\}$, such that if any two
edges $(v, v_j) \in E$ and $(v, v_k) \in E$ with $v_j, v_k \in V_2$,
$v\in V_1$, and $j < k$, then $(v, v_l) \in E$ for all $j \leq l
\leq k$. In other words, for any vertex $v\in V_1$, the subset of
vertices in $V_2$ connected to $v$ forms an interval on the linear
ordering of $V_2$. For any $v\in V_1$, suppose the subset of vertices
in $V_2$ connected to $v$ is $\{v_j,v_{j+1},\ldots,v_k\}$; then we denote
$begin(v,G)=j$ and $end(v,G)=k$. Although $E$ may have $O(n^2)$
edges, it can be represented implicitly by specifying $begin(v,G)$ and
$end(v,G)$ for each $v\in V_1$.  A {\em vertex insertion} on $G$ is
to insert a vertex $v$ into $V_1$ with an edge interval
$[begin(v,G),end(v,G)]$ and implicitly connect $v$ to every $v_i\in
V_2$ with $begin(v,G)\leq i\leq end(v,G)$. Similarly, a {\em vertex
deletion} on $G$ is to delete a vertex $v$ from $V_1$ as well as all its
adjacent edges.

A bipartite graph $G= (V_1, V_2, E)$ is {\it circular convex} on
the vertex set $V_2$ if there is a circular ordering on $V_2$ such
that for each vertex $v \in V_1$, the subset of vertices in $V_2$
connected to $v$ forms a circular-arc interval on that ordering. Precisely,
suppose such a {\em clockwise} circular ordering of $V_2$ is $v_0,v_1,\ldots,v_{n-1}$.
For any two edges $(v, v_j) \in E$ and $(v, v_k) \in E$ with
$v_j, v_k \in V_2$, $v\in V_1$, and $j< k$, either
$(v, v_l) \in E$ for all $j \leq l \leq k$, or $(v, v_l) \in E$ for all $k
\leq l \leq n-1$ and $(v, v_l) \in E$ for all $0 \leq l \leq j$. For
each $v\in V_1$, suppose the vertices of $V_2$ connected to $v$ are
from $v_j$ to $v_k$ clockwise on the ordering, then $begin(v,G)$ and
$end(v,G)$ are defined to be $j$ and $k$, respectively.
Vertex insertions and deletions on $G$ are defined similarly.

%$begin(v,G)$ and $end(v,G)$ are defined as follows.
%Suppose $v_0,v_1,\ldots,v_{n-1}$ of $V_2$ are put on
%the boundary of a circle such that when
%we visit $v_0,v_1,\ldots,v_{n-1}, v_0$ in this order, we traverse the boundary
%of the circle clockwise.  Then as we start at $begin(v,G)$ and scan the
%edge interval $[begin(v,G),end(v,G)]$ of $v$,
%we traverse the circle boundary clockwise.

A maximum matching in a convex bipartite graph can be found in
$O(n)$ time \cite{ref:GabowA85,ref:LipskiEf81,ref:SteinerA96}. The
same time bound holds for a circular convex bipartite graph
\cite{ref:LiangCi95}. Brotal {\em et al.}~\cite{ref:BrodalDy07} designed
a data structure for dynamically maintaining the maximum matchings of
a convex bipartite graph that can support each vertex insertion or
deletion in $O(\log^2 n)$ amortized time. For circular convex
bipartite graphs, however, to our best knowledge, we are not aware of
any previous work on dynamically maintaining their maximum matchings.

The main idea of our algorithm for the decision version of the
min-max problem is as follows. First, we model the problem as
finding the maximum matchings in a sequence of $O(n)$ circular convex
bipartite graphs, which is further modeled as dynamically
maintaining the maximum matching of a circular convex bipartite
graph under a sequence of $O(n)$ vertex insertion and deletion operations.
Second, we develop an approach for solving the latter problem.
Specifically, we show that the maximum matching of a circular convex
bipartite graph of $O(n)$ vertices can be dynamically maintained in $O(\log^2 n)$ time
(in the worst case) for each vertex insertion or deletion. Note that
this result is of independent interest.

%Our approach for this can be viewed as a combination of the dynamic
%data structure in \cite{ref:BrodalDy07} for maintaining the maximum
%matching in a convex bipartite graph and the linear time algorithm
%in \cite{ref:LiangCi95} for computing a maximum matching in a
%circular convex bipartite graph. Additionally, in our particular
%problem setting, each insertion or deletion can be handled in
%$O(\log^2 n)$ time in the worst case. Consequently, the decision
%problem of the min-max problem on a circle is solvable in $O(n\log^2 n)$ time.

In the following, we first present the problem modeling and then
give our algorithm for dynamically maintaining the maximum matching of
a circular convex bipartite graph.

\subsection{The Problem Modeling}
\label{subsec-model}

Recall that in the decision version of the min-max problem,
our goal is to determine whether $\lambda_C\leq
\lambda$. Let $P$ be an arbitrary regular $n$-gon with its vertices
$P_0,P_1,\ldots,P_{n-1}$ ordered clockwise on $\partial C$. We first
consider the following sub-problem: Determine whether we can move all sensors to
the vertices of $P$ such that the maximum distance traveled by the
sensors is at most $\lambda$. Let $G_P$ be the bipartite graph
between the sensors $A_0,\ldots,A_{n-1}$ and the vertices of $P$,
such that a sensor $A_i$ is connected to a vertex $P_j$ in $G_P$ if and only
if $|A_iP_j|\leq \lambda$.  The next lemma is immediate.
\begin{Lem}\label{lem:cir-conv}
The bipartite graph $G_P$ is circular convex.
\end{Lem}
\begin{proof}
This simply follows from the fact that the boundary of any circle of radius $\lambda$
can intersect $\partial C$ at most twice.
\end{proof}

To solve the above sub-problem, it suffices to compute a maximum matching $M$
in the circular convex bipartite graph $G_P$ (by using the algorithm
in \cite{ref:LiangCi95}). If $M$ is a
perfect matching, then the answer to the sub-problem is ``yes";
otherwise, the answer is ``no". Thus, the sub-problem can be solved in
$O(n)$ time (note that the graph $G_P$ can be constructed implicitly
in $O(n)$ time, after $O(n\log n)$ time preprocessing).
If the answer to the sub-problem is ``yes", then we say
that $P$ is {\em feasible} with respect to the value $\lambda$.

If $P$ is feasible, then clearly $\lambda_C\leq \lambda$.
If $P$ is not feasible, however, $\lambda_C> \lambda$ does not necessarily hold,
because $P$ may not be positioned ``right" (i.e., $P$ may not be the regular
$n$-gon in an optimal solution of
the optimization version of the min-max problem). To further
decide whether $\lambda_C\leq \lambda$, our strategy is to rotate
$P$ clockwise on $\partial C$ by an arc distance at most $2\pi/n$.
Since the perimeter of $C$ is $2\pi$, the arc distance
between any two neighboring vertices of $P$ is $2\pi/n$. A simple
yet critical observation is that $\lambda_C\leq \lambda$ if and only
if during the rotation of $P$, there is a moment (called a {\em
feasible moment}) at which $P$
becomes feasible with respect to $\lambda$. %Note that after we
%rotate an arc distance $2\pi/n$, the formation of $P$ is the same as
%that before rotation, and that is why we only rotate $2\pi/n$.
Thus, our task is to determine whether a feasible moment exists
during the rotation of $P$.

Consider the graph $G_P$. For each sensor $A_i$, denote by
$E(A_i)=\{P_j,P_{j+1},\ldots,P_k\}$ the subset of vertices of $P$ connected to
$A_i$ in $G_P$, where the indices of the vertices of $P$ are
taken as module by $n$. We assume that $E(A_i)$ does not contain all
vertices of $P$ (otherwise, it is trivial). Since the arc distance
from $P_{j-1}$ to $P_j$ is $2\pi/n$, during the (clockwise) rotation
of $P$, there must be a moment after which $P_{j-1}$ becomes connected to
$A_i$, and we say that $P_{j-1}$ is {\em added} to $E(A_i)$; similarly,
there must be a moment after which $P_{k}$ becomes disconnected to $A_i$,
and we say that $P_{k}$ is {\em removed} from $E(A_i)$.
Note that these are the moments when the edges of $A_i$ (and thus the graph $G_P$)
are changed due to the rotation of $P$.  Also, note that during the rotation, all
vertices in $E(A_i)\setminus \{P_k\}$ remain connected to $A_i$ and
all vertices in $P\setminus \{E(A_i) \cup \{P_{j-1}\}\}$ remain
disconnected to $A_i$. Hence throughout this rotation,
there are totally $n$ additions and $n$ removals on the graph $G_P$.
If we sort all these additions and removals based on the time
moments when they occur, then we obtain a sequence of $2n$ circular
convex bipartite graphs, and determining whether there exists a
feasible moment is equivalent to determining whether there is a graph
in this sequence that has a perfect matching. With the $O(n)$ time maximum matching
algorithm for circular convex bipartite graphs of $n$ vertices in
\cite{ref:LiangCi95}, a straightforward solution for determining whether there
is a feasible moment would take $O(n^2)$ time.

To obtain a faster algorithm, we further model the problem as follows.
%dynamically maintaining a maximum matching in a circular
%convex bipartite graph with a sequence of $2n$ vertex insertions and
%$2n$ vertex deletions. To this end,
Consider the addition of $P_{j-1}$ to $E(A_i)$. This can be done
by first deleting the vertex of $G_P$ corresponding to $A_i$ and
then inserting a new vertex corresponding to $A_i$ with its edges connecting
to the vertices in $\{P_{j-1}\}\cup E(A_i)$.
%Note that this vertex
%insertion is specified by an interval $[j-1,k]$ and we do not need to
%explicitly connect $A_i$ to all vertices of $P$ from $P_{j-1}$ to $P_k$.
The removal of $P_k$ from $E(A_i)$ can be handled similarly. Thus,
each addition or removal on $E(A_i)$ can be transformed to one vertex
deletion and one vertex insertion on $G_P$. If we sort all vertex
updates (i.e., insertions and deletions) by the time moments when
they occur, then the problem of determining whether there is a feasible moment is
transformed to determining whether there exists a perfect matching
in a sequence of vertex updates on the graph $G_P$. In other words, we need to
dynamically maintain the maximum matching in a circular convex
bipartite graph to support a sequence of $2n$ vertex insertions and
$2n$ vertex deletions.  This problem is handled in the next
subsection.

\subsection{Dynamic Maximum Matching in a Circular Convex Bipartite Graph}

In this subsection, we consider the problem of dynamically maintaining the
maximum matching in a circular convex bipartite graph to support vertex
insertions and deletions. We treat all vertex updates in an online fashion.

Let $G = (V_1, V_2, E)$ with $|V_1|= O(n)$ and $|V_2|=O(n)$ be a circular convex
bipartite graph on the vertex set $V_2$, i.e., the vertices of $V_2$
connected to each vertex in $V_1$ form a circular-arc interval on the
sequence of the vertex indices of $V_2$. Suppose
$V_2=\{v_0,v_1,\ldots,v_{n-1}\}$ is ordered clockwise. Recall that a
vertex insertion on $G$ is to insert a vertex $v$ into $V_1$ with an edge interval
$[begin(v,G),end(v,G)]$ such that $v$ is (implicitly) connected to all
vertices of $V_2$ from $begin(v,G)$ clockwise to $end(v,G)$.
A vertex deletion is to delete a
vertex $v$ from $V_1$ and all its adjacent edges (implicitly).  Our
task is to design an algorithm for maintaining the maximum matching of
$G$ to support such update operations (i.e., vertex insertions and
deletions) efficiently. Below, we present an algorithm with an $O(\log^2n)$ time per
update operation.

Our approach can be viewed as a combination of the data structure in
\cite{ref:BrodalDy07} for dynamically maintaining the maximum
matching in a convex bipartite graph and the linear time algorithm
in \cite{ref:LiangCi95} for computing a maximum matching in a
circular convex bipartite graph. We refer to them as the BGHK data
structure \cite{ref:BrodalDy07} and the LB algorithm
\cite{ref:LiangCi95}, respectively. We first briefly describe
the BGHK data structure and the LB algorithm.

The BGHK data structure \cite{ref:BrodalDy07} is a binary tree $T$,
and each node of $T$ maintains a balanced binary tree. This data structure
can be constructed in $O(n\log^2 n)$ time and can
support each vertex insertion or deletion in $O(\log^2 n)$ amortized time.
%The tree $T$ may become unbalanced after update operations and
%that is why the amortized analysis comes.
Consider a vertex insertion, i.e., inserting a vertex $v$ into $V_1$. Let
$M'$ (resp., $M$) be the maximum matching in the graph before
(resp., after) the insertion. Let $|M|$ denote the number of matched
pairs in $M$.  After the data structure is updated (in $O(\log^2 n)$
amortized time), the value $|M|$ can be reported in $O(1)$ time
and $M$ can be reported in $O(|M|)$ time. We can also determine in
$O(1)$ time whether $v$ is matched in $M$. Further, if another
vertex $v'\in V_1$ was matched in $M'$ but is not matched in $M$,
then it is easy to see that $v$ must be matched in $M$. When this
case occurs, we say that $v$ {\em replaces} $v'$ and $v'$ is called
the {\em replacement}, and the data structure is able to report the
replacement in $O(1)$ time. Note that as shown in
\cite{ref:BrodalDy07}, although an update on the graph can cause
dramatic changes on the maximum matching, the sets of the
matched vertices in $V_1$ (and $V_2$) can change by at most one
vertex. Thus, there is at most one such replacement $v'$.  Similarly,
consider deleting a vertex $v$ from $V_1$. After the data structure
is updated, the value $|M|$ can be reported in $O(1)$ time and $M$
can be reported in $O(|M|)$ time. We can also find out whether $v$ was
matched in $M'$ in $O(1)$ time. If a vertex $v'\in V_1$ was not
matched in $M'$ but is matched in $M$,
then it is easy to see that $v$ must be matched in $M'$. When this
case occurs, we say $v'$ is the {\em
supplement}, which can be determined in $O(1)$ time.

The LB algorithm \cite{ref:LiangCi95} finds a maximum matching in
a circular convex bipartite graph $G=(V_1,V_2,E)$ by reducing the
problem to two sub-problems of computing the maximum matchings in two
convex bipartite graphs $G_1$ and $G_2$. Some details are summarized
below. For any vertex $v\in V_1$, if $begin(v,G)\leq end(v,G)$, then
$v$ is called a {\em non-boundary} vertex. Otherwise, $v$ is a {\em
boundary} vertex; the edges connecting $v$ to
$v_{begin(v,G)},v_{begin(v,G)+1},\ldots,v_{n-1}$ in $V_2$ are called {\em lower
edges}, and the other edges connecting $v$ are {\em upper edges}.
Based on the graph $G$, a convex bipartite graph
$G_1=(V_1,V_2,E_1)$ is defined as follows. Both its vertex sets are
the same as those in $G$. For each vertex $v\in V_1$ in $G$,
$begin(v,G_1)=begin(v,G)$; if $v$ is a non-boundary vertex,
then $end(v,G_1)=end(v,G)$, and otherwise $end(v,G_1)=n-1+end(v,G)$ (note
that this value of $end(v,G_1)$ is used only for comparison in the algorithm
although there are not so many vertices in $V_2$). The LB
algorithm has two main steps. The first step is to compute a
maximum matching in $G_1$, which can be done in $O(n)$ time
\cite{ref:GabowA85,ref:LipskiEf81,ref:SteinerA96}. Let $M(G_1)$ be
the maximum matching of $G_1$. Next, another convex
bipartite graph $G_2=(V_1,V_2,E_2)$ is defined based on $M(G_1)$ and $G$, as
follows. Both its vertex sets are the same as those in $G$. For each
non-boundary vertex $v\in V_1$ in $G$, $begin(v,G_2)=begin(v,G)$ and
$end(v,G_2)=end(v,G)$. For each boundary vertex $v\in V_1$ in $G$,
there are two cases: If $v$ is matched in $M(G_1)$, then
$begin(v,G_2)=begin(v,G)$ and $end(v,G_2)=n-1$; otherwise,
$begin(v,G_2)=0$ and $end(v,G_2)=end(v,G)$. The second step of the LB
algorithm is to compute a maximum matching in $G_2$ (in $O(n)$
time), denoted by $M(G_2)$. It was shown in \cite{ref:LiangCi95}
that $M(G_2)$ is also a maximum matching of the original graph $G$.

%Note that in \cite{ref:LiangCi95}, the second step of the algorithm does
%not consider those non-boundary vertices in $V_1$ that are unmatched
%in $M(G_1)$. However, in our description above, those vertices are
%still considered in $G_2$. It is easy to show that those vertices will never be
%matched in $M(G_2)$. So considering those vertices in $G_2$ will not
%affect the correctness or the running time of the LB algorithm.
%We slightly modify the LB algorithm as discussed above is due to the
%ease of the exposition for our dynamic algorithm given in the sequel
%although a slightly different (but more tedious) dynamic algorithm can also be given by
%strictly following the original LB algorithm in \cite{ref:LiangCi95}.

We now discuss our algorithm for dynamically maintaining a
maximum matching in the circular convex bipartite graph $G$. As
preprocessing, we first run the LB algorithm on $G$, after which
both the convex bipartite graphs $G_1$ and $G_2$ of $G$ are
available. We then build two BGHK data structures for $G_1$ and
$G_2$, denoted by $T(G_1)$ and $T(G_2)$, respectively, for
maintaining their maximum matchings. This completes the preprocessing,
which takes $O(n\log^2 n)$ time. In the following, we discuss how to perform
vertex insertions and deletions.

Consider a vertex insertion, i.e., inserting a vertex $v$ into $V_1$ with the edge
interval $[begin(v,G)$, $end(v,G)]$. To perform this insertion, intuitively, we need to
update the two BGHK data structures $T(G_1)$ and $T(G_2)$ in a way that mimics some
behavior of the LB algorithm. Specifically, we first insert $v$ into the graph
$G_1$ by updating $T(G_1)$. Based on the results on $G_1$
(e.g., whether there is a replacement) and the behavior of the LB algorithm,
we modify $G_2$ by updating $T(G_2)$ accordingly.
In this way, the maximum matching maintained by
$T(G_2)$ is the maximum matching of $G$ after the insertion.
The details are given below.

Let $G_1'$ and $G_2'$ be the two graphs
that would be produced by running the LB algorithm on $G$ with the new vertex $v$
(and its adjacent edges). Let $M(G_1), M(G_2), M(G_1')$, and $M(G_2')$ be the maximum
matchings of $G_1, G_2, G_1'$, and $G_2'$, respectively.
Depending on whether $v$ is a boundary
vertex, there are two main cases.

\begin{itemize}
\item
If $v$ is a non-boundary vertex (i.e., $begin(v,G)\leq end(v,G)$), then
$G_1'$ can be obtained by inserting $v$ into $G_1$. Hence we insert $v$ into
$T(G_1)$. Depending on whether there is a replacement, there are two
cases.

\begin{itemize}
\item
If no replacement, then $G_2'$ can be obtained by inserting $v$
into $G_2$. Thus, we simply insert $v$ into $T(G_2)$ and we
are done.

\item
Otherwise, let $v'$ be the replacement. So $v'$ was matched in
$M(G_1)$ but is not matched in $M(G_1')$.
Depending on whether $v'$ is a boundary vertex, there are two
subcases.

\begin{itemize}
\item
If $v'$ is a non-boundary vertex, then again, $G_2'$ can be obtained
by inserting $v$ into $G_2$. We thus insert $v$ into $T(G_2)$ and we
are done.

\item
If $v'$ is a boundary vertex, then
since $v'$ was matched in $M(G_1)$, according to the LB algorithm,
$v'$ with the edge interval $[begin(v',G),n-1]$ is in $G_2$.
After the insertion of $v$ into $G_1$, $v'$ is not matched in $M(G'_1)$.
Thus, according to the LB algorithm, $G_2'$ can be obtained by
deleting $v'$ (with the edge interval $[begin(v',G),n-1]$) from $G_2$, inserting
$v'$ with the edge interval $[0,end(v',G)]$ into $G_2$, and finally inserting $v$ into $G_2$.

In summary, for this subcase, we delete $v'$ (with the edge interval
$[begin(v',G),n-1]$) from $T(G_2)$ and insert
$v'$ with the edge interval $[0,end(v',G)]$ into $T(G_2)$.
Finally, we insert $v$ into $T(G_2)$, and we are done.
\end{itemize}
\end{itemize}

\item
If $v$ is a boundary vertex (i.e., $begin(v,G)> end(v,G)$),
then according to the LB algorithm, $G_1'$ can be obtained by
inserting $v$ with the edge interval $[begin(v,G),n-1+end(v,G)]$ into $G_1$.
Thus we insert $v$ with the edge interval $[begin(v,G),n-1+end(v,G)]$ into $T(G_1)$.
Depending on whether there is a replacement, there are two
cases.

\begin{itemize}
\item
If no replacement, then depending on whether $v$ is matched in $M(G_1')$,
there are two subcases.

\begin{itemize}
\item
If $v$ is matched, then according to the LB algorithm, $G_2'$ can be
obtained by inserting $v$ with the edge interval $[begin(v,G),n-1]$ into $G_2$. Thus,
we insert $v$ with the edge interval $[begin(v,G),n-1]$ into $T(G_2)$, and we are
done.

\item
If $v$ is not matched, then according to the LB algorithm, $G_2'$ can be
obtained by inserting $v$ with the edge interval $[0, end(v,G)]$ into $G_2$.
We thus insert $v$ with the edge interval $[0, end(v,G)]$ into $T(G_2)$, and we are done.

\end{itemize}

\item
Otherwise, there is a replacement $v'$. So $v'$
was matched in $M(G_1)$ but is not matched in $M(G_1')$, and $v$ is
matched in $M(G_1')$. Depending on whether $v'$ is a boundary vertex, there are two
subcases.

\begin{itemize}
\item
If $v'$ is a non-boundary vertex, then since $v$ is
matched in $M(G_1')$, $G_2'$ can be
obtained by inserting $v$ with the edge interval $[begin(v,G),n-1]$ into $G_2$.
We thus insert $v$ with the edge interval $[begin(v,G),n-1]$ into $T(G_2)$.
%and we are done.

\item
If $v'$ is a boundary vertex, then according to the LB algorithm,
$G_2'$ is the graph obtained by deleting $v'$ (with the edge interval
$[begin(v',G),n-1]$) from $G_2$, inserting $v'$ with the edge interval $[0,end(v',G)]$ into
$G_2$, and finally inserting $v$ with the edge interval $[begin(v,G),n-1]$ into $G_2$.

Thus, we delete $v'$ (with the edge interval $[begin(v',G),n-1]$) from $T(G_2)$, and insert
$v'$ with the edge interval $[0,end(v',G)]$ into $T(G_2)$.
Finally, we insert $v$ with the edge interval $[begin(v,G),n-1]$ into $T(G_2)$.
%and we are done.
\end{itemize}

\end{itemize}
\end{itemize}

%In addition, the graph $G$ is (implicitly) updated by inserting $v$
%and all its adjacent edges.
This completes the description of our procedure for handling a vertex insertion.

Next, consider a vertex deletion, i.e., deleting a vertex $v$ from $V_1$ of
$G$. Our procedure for this operation proceeds in a manner symmetric to the insertion
procedure, and we briefly discuss it below. Define the two graphs $G_1'$ and $G_2'$ similarly as above.

\begin{itemize}
\item If $v$ is a non-boundary vertex, then we delete $v$ from
$T(G_1)$. If no supplement, then we delete $v$ from $T(G_2)$ and we
are done. Otherwise, let $v'$ be the supplement. So $v'$ was not
matched in $M(G_1)$ but is matched in $M(G_1')$.
Depending on whether $v'$ is a boundary vertex, there are two cases.

\begin{itemize}
\item If $v'$ is a non-boundary vertex, then we delete $v$ from
$T(G_2)$ and we are done.

\item If $v'$ is a boundary vertex, then we delete $v'$ with the edge interval
$[0,end(v',G)]$ from $T(G_2)$ and insert $v'$ with the edge interval
$[begin(v',G),n-1]$ into
$T(G_2)$. Finally, delete $v$ from $T(G_2)$, and we are done.

\end{itemize}

\item If $v$ is a boundary vertex, then we delete $v$ (with the edge interval
$[begin(v,G),n-1+end(v,G)]$) from $T(G_1)$. Depending on whether there is
a supplement, there are two cases.

\begin{itemize}
\item
If no supplement, then depending on whether $v$ was matched in
$M(G_1)$, there are two subcases. If $v$ was matched, then we
delete $v$ (with the edge interval $[begin(v,G),n-1]$) from $T(G_2)$;
otherwise, we delete $v$ (with the edge interval $[0,end(v,G)]$) from $T(G_2)$.
%We are done.

\item Otherwise, let $v'$ be the supplement.
So $v'$ was not matched in $M(G_1)$ but is matched in $M(G_1')$, and
$v$ was matched in $M(G_1)$.
Since $v$ was matched in $M(G_1)$, according to the LB algorithm,
$G_2$ contains $v$ with the edge interval $[begin(v,G),n-1]$.
If $v'$ is a non-boundary vertex, then we delete $v$ (with the edge interval
$[begin(v,G),n-1]$) from $T(G_2)$ and we are done.
Otherwise, since $v'$ was not matched in $M(G_1)$, according to the LB
algorithm, $G_2$ contains $v'$ with the edge interval $[0,end(v',G)]$;
since $v'$ is matched in $M(G_1')$, according to the LB
algorithm, $G_2'$ should contain $v'$ with the edge interval
$[begin(v',G),n-1]$. Therefore,
we delete $v'$ (with the edge interval $[0,end(v',G)]$) from $T(G_2)$,
insert $v'$ with the edge interval $[begin(v',G),n-1]$ into $T(G_2)$, and finally
delete $v$ (with the edge interval $[begin(v,G),n-1]$) from $T(G_2)$.
%We are done.
\end{itemize}
\end{itemize}

%In addition, the graph $G$ is (implicitly) updated by deleting $v$
%and all its adjacent edges.
This completes the description of our vertex deletion procedure.

As shown in Subsection \ref{subsec-model}, the decision version of the min-max problem
can be transformed to the problem of dynamically maintaining the maximum matching
in a circular convex bipartite graph subject to a sequence of vertex insertions and deletions.
Hence, the correctness of our algorithm for the decision version hinges on the correctness
of our dynamic maximum matching algorithm for circular convex bipartite graphs.
Yet, the correctness of our (online) dynamic maximum matching algorithm for circular
convex bipartite graphs can be seen quite easily. This is because our procedures
for performing vertex insertions and deletions are both based on the fact that
they simply mimic the behavior of the LB algorithm (while
implementing their processing by the means of the BGHK data structures).

For the running time of our algorithm, each update operation involves
at most two vertex insertions and two vertex deletions on $T(G_1)$ and
$T(G_2)$, each of which takes $O(\log^2n)$ amortized time
\cite{ref:BrodalDy07}; thus, it takes $O(\log^2n)$ amortized time in total.
Actually, the BGHK data structure in \cite{ref:BrodalDy07} supports
vertex insertions and deletions not only on $V_1$ but also on $V_2$.
Inserting vertices on $V_2$ may make the tree unbalanced, and that is why its
running time is amortized. However, if vertices are inserted only on
$V_1$, then the tree will never become unbalanced and thus each update takes $O(\log^2n)$
time in the worst case. In our
problem formulation, the vertex updates indeed are only on $V_1$.
Denote by $M(G)$ the maximum matching in $G$. We then have the following result.

\begin{Theo}\label{theo:10}
A data structure on a circular convex bipartite graph
$G=(V_1,V_2,E)$ can be built in $O(n\log^2 n)$ time for
maintaining its maximum matching $M(G)$ so that each online
vertex insertion or deletion on $V_1$ can be done in $O(\log^2 n)$ time in the
worst case. After each update operation, $|M(G)|$ can be reported in
$O(1)$ time and $M(G)$ can be reported in $O(|M(G)|)$ time.
\end{Theo}

Since the decision version of the min-max problem
has been reduced to dynamically maintaining the maximum matching
in a circular convex bipartite graph
under a sequence of $2n$ vertex insertions and $2n$ vertex decisions,
we solve the dynamic maximum problem as follows.  After each update
operation, we check whether $|M(G)|=n$, and if this is true, then we report
$\lambda_C\leq \lambda$ and halt the algorithm. If all $4n$ updates have
been processed but it is always $|M(G)|<n$, then we report
$\lambda_C>\lambda$.  Based on Theorem \ref{theo:10}, we have the result below.

\begin{Theo}\label{theo:20}
Given a value $\lambda$, we can determine whether
$\lambda_C\leq\lambda$ in $O(n\log^2 n)$ time for the decision version
of the min-max problem.
\end{Theo}

\section{The Optimization Version of the Min-max Problem}
\label{sec:optimization}

In this section, we consider the optimization version of the min-max
problem, and present an $O(n\log^3 n)$ time algorithm for it.
The main task is to compute the value $\lambda_C$.

Let $o$ be the center of $C$. For simplicity of discussion, we assume that no sensor
lies at $o$.
Denote by $X_{i}$ and $Y_i$ the two points on $\partial C$ which are closest
and farthest to each sensor $A_{i}$, respectively.
Clearly, $X_{i}$ and $Y_i$ are the two intersection
points of $\partial C$ with the line passing through $A_{i}$ and the center $o$
of $C$ (see Figure 1(a)). The lemma below has been proved in
\cite{ref:TanNe10}, and for self-containment of this paper, we include that proof in
Appendix \ref{app:lemmaproofs}.

\begin{figure}[t]
\begin{center}
\includegraphics[height=1.3in,clip]{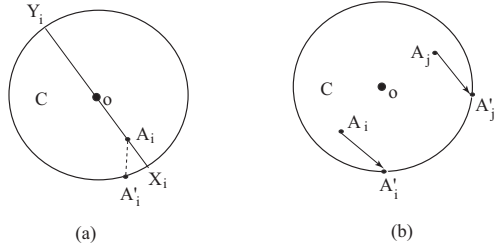}
\caption{(a) The points $X_i$ and $Y_i$ on $\partial C$ for $A_i$;
(b) $|A_{i} A'_{i}| = |A_{j} A'_{j}|$.}
\end{center}
\vspace*{-0.2in}
\end{figure}

\begin{Lem}\label{lem:10}\cite{ref:TanNe10}
Suppose an optimal solution for the min-max optimization problem is achieved
with $\lambda_C = |A_{i} A'_{i}|$ for some $i\in \{0,\ldots,
n-1\}$. Then either $A'_{i}$ is the point $X_{i}$, or there is
another sensor $A_{j}$ ($j \neq i$) such that $\lambda_C = |A_{j}
A'_{j}|$ also holds. In the latter case, any slight rotation of the
regular $n$-gon that achieves $\lambda_C$ in either direction
causes the value of $\lambda_C$ to increase (i.e., it makes one of the two
distances $|A_i A'_i|$ and $|A_j A'_j|$ increase and the other one decrease).
%monotonically increases one of the two distances $|A_i A'_i|$ and $|A_j A'_j|$,
%and decreases the other.
\end{Lem}

The points on $\partial C$ satisfying the conditions specified in
Lemma \ref{lem:10}
%clearly contribute to the
may be considered as those defining candidate values for $\lambda_C$,
i.e., they can be considered as some vertices of possible regular
$n$-gons on $\partial C$ in an optimal solution. The points $X_{h}$
of all sensors $A_h$ ($0 \leq h \leq n-1$) can be easily determined.
Define $D_1=\cup_{h=0}^{n-1}\{|A_hX_h|\}$, which can be computed in
$O(n)$ time. But, the challenging task is to handle all the pairs
$(A_i, A_j)$ ($i \neq j$) such that the distance from $A_i$ to a
vertex of a regular $n$-gon is equal to the distance from $A_j$ to
another vertex of that $n$-gon and a slight rotation of the $n$-gon
in either direction monotonically increases one of these two
distances but decreases the other. We refer to such distances as the
{\it critical equal distances}. Denote by $D_2$ the set of all critical equal
distances. Let $D=D_1\cup D_2$. By Lemma \ref{lem:10}, $\lambda_C\in
D$. Thus, if $D$ is somehow available, then $\lambda_C$ can be
determined by using our algorithm in Theorem \ref{theo:20} in a binary search process. Since
$D_1$ is readily available, the key is to deal with $D_2$
efficiently. An easy observation is $\max_{0\leq h\leq n-1}|A_hX_h|\leq \lambda_C$.
We can use the algorithm in Theorem \ref{theo:20} to check whether
$\lambda_C\leq \max_{0\leq h\leq n-1}|A_hX_h|$, after which we know
whether $\lambda_C=\max_{0\leq h\leq n-1}|A_hX_h|$. Below, we assume
$\max_{0\leq h\leq n-1}|A_hX_h|< \lambda_C$ (otherwise, we are
done). Thus, we only need to focus on finding $\lambda_C$ from the set $D_2$.

%Note that the general case can be handled
%without affecting the performance of the algorithm but the discussion
%would be more tedious.

It has been shown in \cite{ref:TanNe10} that $|D_2|=O(n^3)$.
Of course, our goal is to avoid an $O(n^3)$ time solution.
To do so, first we determine a subset $D_2'$ of $D_2$ such that
$\lambda_C\in D_2'$ but with $|D_2'|=O(n^2)$.
%The set $D_2'$ can be
%computed in $O(n^2)$ time, and thus, by using our algorithm in
%Theorem \ref{theo:20} together with the linear time selection
%algorithm \cite{ref:CLRS01}, we find $\lambda_C$ in $O(n^2)$ time.
%To reduce the running time,
Furthermore, we do not compute $D_2'$ explicitly. Specifically, our
idea is as follows. We show that the elements of $D_2'$ are the
$y$-coordinates of a subset of intersection points among a set $F$ of $O(n)$
functional curves in the plane such that each curve is $x$-monotone and any two
such curves intersect in at most one point at which the two curves cross
each other. (Such a set of curves is sometimes referred to as
{\em pseudolines} in the literature.) Let $\calA_F$ be the arrangement
of $F$ and $|\calA_F|$ be the number of vertices of $\calA_F$. Without
computing $\calA_F$ explicitly, we will generalize the techniques in
\cite{ref:ColeAn89} to compute the $k$-th highest vertex of
$\calA_F$ for any integer $k$ with $1\leq k\leq |\calA_F|$ in
$O(n\log^2 n)$ time. Consequently, with Theorem \ref{theo:20}, the
value $\lambda_C$ can be computed in $O(n\log^3 n)$ time. The
details are given below.

%Then, our algorithm for
%computing $\lambda_C$ works as follows. First, compute the
%$n^2/2$-th topmost vertex of $\calA_F$, and let its $y$-coordinate
%be $\lambda_m$. By using the algorithm in Theorem \ref{theo:20},
%determine whether $\lambda_C\leq \lambda_m$, after which one half of
%the vertices in $\calA_F$ can be pruned away. We then apply this
%process recursively on the remaining vertices in $\calA_F$ until
%$\lambda_C$ is found. With Theorem \ref{theo:20}, the total time of
%this recursive procedure is clearly $O(n\log^3n)$. The details are
%given below.

%First, select the median $\lambda_m$ of the elements in $D_2'\cup
%D_1$. By using the algorithm in Theorem \ref{theo:20}, determine
%whether $\lambda_C\leq \lambda_m$, after which one half of the
%elements in $D_2'\cup D_1$ can be pruned away. We then apply this
%process recursively on the remaining elements in $D_2'\cup D_1$,
%until $\lambda_C$ is found. The total time of this recursive
%processing is clearly $O(n^2)$.

%The rest of this section is devoted to computing such a set $D_2'$.

Let $P$ be an arbitrary regular $n$-gon with its vertices $P_0$,
$P_1$, $\ldots$, $P_{n-1}$ clockwise on $\partial C$. Suppose the
distances of all the pairs between a sensor and a vertex of $P$ are $d_1 \leq
d_2 \leq \cdots \leq d_{n^{2}}$ in sorted order. Let $d_0=0$.
Clearly, $d_0<\lambda_C\leq d_{n^2}$ (the case of $\lambda_C=0$ is
trivial). Hence, there exists an integer $k$ with $0\leq k<n^2$ such
that $\lambda_C\in (d_k,d_{k+1}]$. One can find $d_k$ and $d_{k+1}$
by first computing all these $n^2$ distances explicitly and then utilizing
our algorithm in Theorem \ref{theo:20} in a binary search process.
But that would take $\Omega(n^2)$
time. In the following lemma, we give a faster procedure without
having to compute these $n^2$ distances explicitly.

\begin{Lem}\label{lem:new20}
The two distances $d_k$ and $d_{k+1}$
can be obtained in $O(n\log^3 n)$ time.
\end{Lem}
\begin{proof}
We apply a technique, called {\em binary search in sorted arrays}
\cite{ref:ChenRe11},
as follows. Given $M$ arrays $A_i$, $1\leq i\leq M$, each containing
$O(N)$ elements in sorted order, the task is to find a certain element $\delta \in
A=\cup_{i=1}^M A_i$. Further, assume that there is a ``black-box"
decision procedure $\Pi$ available, such that given any value $a$,
$\Pi$ reports $a\leq \delta$ or $a>\delta$ in $O(T)$ time. An
algorithm is given in \cite{ref:ChenRe11} to find the sought element
$\delta$ in $A=\cup_{i=1}^MA_i$ in $O((M+T)\log (NM))$
time. We use this technique to find $d_k$ and $d_{k+1}$, as follows.

Consider a sensor $A_i$. Let $S(A_i)$ be the set of distances
between $A_i$ and all vertices of $P$. In $O(\log n)$ time, we can
implicitly partition $S(A_i)$ into two sorted arrays in the
following way. By binary search, we can determine an index $j$ such
that $X_i$ lies on the arc of $\partial C$ from $P_j$ to $P_{j+1}$
clockwise (the indices are taken as module by $n$). Recall that
$X_i$ is the point on $\partial C$ closest to $A_i$. If a vertex of
$P$ is on $X_i$, then define $j$ to be the index of that vertex.
Similarly, we can determine an index $h$ such that $Y_i$ (i.e., the
farthest point on $\partial C$ to $A_i$) lies on the arc from $P_h$
to $P_{h+1}$ clockwise. If a vertex of $P$ is on $Y_i$, then define
$h$ to be the index of that vertex. Both $j$ and $h$ can be
determined in $O(\log n)$ time, after which we implicitly
partition $S(A_i)$ into two sorted arrays: One array consists of all
distances from $A_i$ to $P_j,P_{j-1},\ldots,P_{h+1}$, and the other
consists of all distances from $A_i$ to
$P_{j+1},P_{j+2},\ldots,P_{h}$ (again, all indices are taken as
module by $n$). Note that both these arrays are sorted increasingly and
each element in them can be obtained in $O(1)$ time by using its index in
the corresponding array.

Thus, we obtain $2n$ sorted arrays (represented implicitly) for all $n$ sensors in
$O(n\log n)$ time, and each array has no more than $n$ elements.
Therefore, by using the technique of binary search in sorted arrays, with
our algorithm in Theorem \ref{theo:20} as the black-box decision
procedure, both $d_k$ and $d_{k+1}$ can be found in $O(n\log^3 n)$
time. The lemma thus follows.
\end{proof}

By applying Lemma \ref{lem:new20}, we have $\lambda_C\in (d_k,d_{k+1}]$.
%After finding $d_k$ and $d_{k+1}$, by the algorithm in Theorem \ref{theo:20},
%we determine whether $\lambda_C\leq d_{k+1}$.
Below, for simplicity of discussion, we assume $\lambda_C\neq d_{k+1}$.
Thus $\lambda_C\in(d_k,d_{k+1})$. Since
$\max_{0\leq h\leq n-1}|A_hX_h|< \lambda_C$,
we redefine $d_k:=\max\{d_k,\max_{0\leq h\leq n-1}|A_hX_h|\}$.
We still have $\lambda_C\in(d_k,d_{k+1})$.
Let $D_2'$ be the set of all
critical equal distances in the range $(d_k,d_{k+1})$. Then $\lambda_C\in D_2'$. We
show below that $|D_2'|=O(n^2)$ and $\lambda_C$ can be found in
$O(n\log^3 n)$ time without computing $D_2'$ explicitly.

%\begin{Lem}\label{lem:20}
%Let $D_2'$ be the set of all critical equal distances in $(d_k,d_{k+1}]$.
%Then, $|D_2'| = O(n^2)$, and the set $D_2'$ can be computed in $O(n^2)$ time.
%\end{Lem}

Suppose we rotate the regular $n$-gon $P=(P_0,P_1,\ldots,P_{n-1})$ on
$\partial C$ clockwise by an arc distance $2\pi/n$ (this is the arc
distance between any two adjacent vertices of $P$). Let
$A_i(P_h(t))$ denote the distance function from a sensor $A_i$ to a
vertex $P_h$ of $P$ with the time parameter $t$ during
the rotation. Clearly, the function $A_i(P_h(t))$ increases or
decreases monotonically, unless the interval of $\partial C$ in
which $P_h$ moves contains the point $X_i$ or $Y_i$; if that interval
contains $X_i$ or $Y_i$, then we can further divide the interval
into two sub-intervals at $X_i$ or $Y_i$, such that
$A_i(P_h(t))$ is monotone in each sub-interval.
%Due to our assumption that no distance of $D_1$ is in $(d_k,d_{k+1}]$, the
%function $A_i(P_h(t))$ either increases or decreases monotonically.
The functions $A_i(P_h(t))$, for all $P_h$'s of $P$, can thus be
put into two sets $S_{i1}$ and $S_{i2}$ such that all functions
in $S_{i1}$ monotonically increase and all functions in $S_{i2}$
monotonically decrease. Let $m=|S_{i1}|$. Then $m\leq n$. Denote by
$d^{i}_1 < d^{i}_2 < \cdots <d^{i}_m$ the sorted sequence of the
initial values of the functions in $S_{i1}$. Also, let $d^{i}_0=0$
and $d^{i}_{m+1}=2$ (recall that the radius of $C$ is $1$).  It is
easy to see that the range $(d_k, d_{k+1})$ obtained in Lemma \ref{lem:new20}
is contained in $[d^i_j, d^i_{j+1}]$ for some $0\leq j \leq m$.
%, or in $[d^i_j, d^i_{j+1}]$ when $j = h -1$.
The same discussion can be made for the distance functions in the
set $S_{i2}$ as well.

%Consider a sensor $A_i$ and its two function sets $S_{i1}$ and
%$S_{i2}$. Suppose $(d_k, d_{k+1}]$ is contained in $[d^i_j,
%d^i_{j+1}]$ for some $0\leq j \leq m$, where $d^i_j$ and $d^i_{j+1}$
%are two consecutive initial values of the functions in $S_{i1}$ as
%discussed above. Let $A_i(P_x(t))$ (resp., $A_i(P_{y}(t))$) be the
%function in $S_{i1}$ whose initial value is $d^i_{j}$ (resp.,
%$d^i_{j+1}$). If $d_{k+1}\neq d^i_{j+1}$, then since we rotate $P$
%by only an arc distance $2\pi/n$, during the rotation of $P$, among
%all distance functions in $S_{i1}$, only $A_i(P_x(t))$ may vary in
%%the range $(d_k, d_{k+1}]$. If $d_{k+1}= d^i_{j+1}$, then during the
%rotation of $P$, both the functions $A_i(P_x(t))$ and
%$A_i(P_{y}(t))$ may vary in the range $(d_k, d_{k+1}]$, and other
%functions in $S_{i1}$ cannot vary in this range. Similar discussion
%can also be made for the distance functions in the set $S_{i2}$. In
%summary, during the rotation of $P$,
Since we rotate $P$ by only an arc distance $2\pi/n$, during the
rotation of $P$, each sensor $A_i$ can have at most two distance
functions (i.e., one decreasing and one increasing) whose values may
vary in the range $(d_k, d_{k+1})$. We can easily identify these at
most $2n$ distance functions for the $n$ sensors in $O(n\log n)$
time. Denote by $F'$ the set of all such distance functions.
Clearly, all critical equal distances in the range $(d_k,d_{k+1})$ can be
generated by the functions in $F'$ during the rotation of $P$. Because every
such distance function either increases or decreases monotonically
during the rotation of $P$, each pair of one increasing function and
one decreasing function can generate at most one critical equal distance
during the rotation. (Note that by Lemma \ref{lem:10}, a
critical equal distance cannot be generated by two increasing functions or
two decreasing functions.) Since $|F'|\leq 2n$, the total number of
critical equal distances in $(d_k,d_{k+1})$ is bounded by $O(n^2)$, i.e.,
$|D_2'|=O(n^2)$. For convenience of discussion, since we are
concerned only with the critical equal distances in $(d_k,d_{k+1})$, for each
function in $F'$, we restrict it to the range $(d_k,d_{k+1})$ only.

Let the time $t$ be the $x$-coordinate and the function values be
the $y$-coordinates of the plane. Then each function in $F'$ defines a curve
segment that lies in the strip of the plane between the two horizontal lines
$y=d_k$ and $y=d_{k+1}$. We refer to a function in
$F'$ and its curve segment interchangeably, i.e., $F'$ is
also a set of curve segments. Clearly, a critical equal distance generated by an
increasing function and a decreasing function is the $y$-coordinate of
the intersection point of the two corresponding curve segments.
Note that every function in $F'$
has a simple mathematical description. Below, we simply assume that
each function in $F'$ is of $O(1)$ complexity. Thus, many
operations on them can each be performed in $O(1)$ time, e.g., computing the
intersection of a decreasing function and an increasing function.

The set $D_2'$ can be computed explicitly in $O(n^2)$ time, after
which $\lambda_C$ can be easily found by binary search. Below, we develop a faster
solution without computing $D_2'$ explicitly, by utilizing
the property that each element of $D_2'$ is the $y$-coordinate of the
intersection point of a decreasing function and an increasing function in
$F'$ and generalizing the techniques in \cite{ref:ColeAn89}.

A slope selection algorithm for a set of points in the plane was given in
\cite{ref:ColeAn89}. We will extend this approach to solve our problem.
The following lemma is needed.

\begin{Lem}\label{lem:30}
For any two increasing (resp., decreasing) functions in $F'$, if the
curve segments defined by them are not identical to each other, then the two curve
segments intersect in at most one point and they cross each other at
their intersection point (if any).
\end{Lem}
\begin{proof}
We only prove the decreasing case. The increasing case can be proved
similarly. Let $A_i(P_a(t))$ and $A_j(P_b(t))$ be two
decreasing curves in $F'$, where $A_i(P_a(t))$ (resp.,
$A_j(P_b(t))$) is the distance function between the sensor $A_i$
(resp., $A_j$) and the vertex $P_a$ (resp., $P_b$) of the regular $n$-gon
$P$, and the two curve segments defined by $A_i(P_a(t))$ and
$A_j(P_b(t))$ are not the same. Since each sensor has at most one
decreasing function in $F'$, we have $A_i\neq A_j$.
We assume that during the (clockwise) rotation of $P$, $A_i(P_a(t))=A_j(P_b(t))$ at the
moment $t=t_1$ and $t_1$ is the first such moment. Below, we prove
that $A_i(P_a(t))=A_j(P_b(t))$ cannot happen again for any $t>t_1$
in the rotation. There are two cases: $P_a=P_b$ and
$P_a\neq P_b$.

For any two points $p$ and $q$, let $l(p,q)$ denote the line passing
through the two points and $\overline{pq}$ denote the line segment with
endpoints $p$ and $q$ whose length is $|pq|$.
Recall that $o$ is the center of the circle $C$. Let $P_a(t_1)$ and
$P_b(t_1)$ be the positions of $P_a$ and $P_b$ at the moment $t_1$,
respectively.

\begin{itemize}
\item
$P_a=P_b$. Clearly, $P_a(t_1)=P_b(t_1)$.
Let $l$ be the perpendicular bisector of the line
segment $\overline{A_iA_j}$. At the moment $t_1$, since
$|A_iP_a(t_1)|=|A_jP_a(t_1)|$, $P_a(t_1)$ is at one of the two intersection points of
$l$ and $\partial C$. Further, since $A_i(P_a(t))$ is a decreasing
function, $P_a(t_1)$ must be on the right side of the line $l(A_i,o)$ if
we walk from $A_i$ to $o$ (see Fig.~\ref{fig:case10}(a)). Similarly, $P_a(t_1)$
must be on the right side of the line $l(A_j,o)$ (going $A_j$ to $o$). Let $z$ be the
other intersection point of $l$ and $\partial C$. It is easy to see that
$z$ is on the left side of either the line $l(A_i,o)$ or the line
$l(A_j,o)$. Note that $d_k\geq \max_{0\leq h\leq n-1}|A_hX_h|$.
Thus, $d_k\geq |A_iX_i|$ and $d_k\geq |A_jX_j|$. During the rotation of $P$,
since both $A_i(P_a(t))$ and $A_j(P_b(t))$ are always larger than
$d_k$, $P_a(t)$ cannot pass any of $X_i$ and $X_j$, and thus $P_a(t)$
cannot arrive to the position $z$ during the rotation. Hence,
$A_i(P_a(t))=A_j(P_b(t))$ cannot happen again after $t_1$.

Further, recall that $t_1$ is the first moment from the beginning of
the rotation with $A_i(P_a(t))=A_j(P_b(t))$. Without loss of
generality, we assume $A_i(P_a(t))<A_j(P_b(t))$ for any time $t<t_1$ (as
the example shown in Fig.~\ref{fig:case10}(a)). It is easy to see
that $A_i(P_a(t))>A_j(P_b(t))$ for any time $t>t_1$, which implies that
the two functions cross each other at their intersection point.

\begin{figure}[t]
\begin{minipage}[t]{\linewidth}
\begin{center}
\includegraphics[totalheight=1.5in]{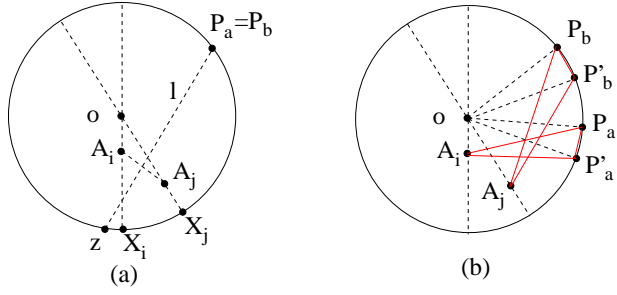}
\caption{\footnotesize Illustrating the proof
of Lemma \ref{lem:30}: (a) $P_a=P_b$; (b) $P_a\neq P_b$.} \label{fig:case10}
\end{center}
\end{minipage}
%\hspace*{0.02in}
\vspace*{-0.10in}
\end{figure}

\item
$P_a\neq P_b$. At the moment $t_1$, we have $|A_iP_a(t_1)|=|A_jP_b(t_1)|$.
Assume to the contrary that at some moment $t_2>t_1$, we also have
$A_i(P_a(t_2))=A_j(P_b(t_2))$. Suppose at the moment $t_2$, $P_a(t_2)$ is at
the position $P_a'$ and $P_b(t_2)$ is at the position $P_b'$ (see
Fig.~\ref{fig:case10}(b)). Then $|A_iP'_a|=|A_jP'_b|$. Since $P_a$
and $P_b$ are rotated simultaneously, the arc distance from $P_a(t_1)$ to
$P_a'$ is equal to the arc distance from $P_b(t_1)$ to $P_b'$, and thus
$|P_a(t_1)P_a'|=|P_b(t_1)P_b'|$. Consider the two triangles $\triangle P_b(t_1)A_jP_b'$
and $\triangle P_a(t_1)A_iP_a'$ (shown with red solid segments
in Fig.~\ref{fig:case10}(b)).
Since $|A_iP_a(t_1)|=|A_jP_b(t_1)|$, $|A_iP'_a|=|A_jP'_b|$, and
$|P_a(t_1)P_a'|=|P_b(t_1)P_b'|$, $\triangle P_b(t_1)A_jP_b'$ is congruent to
$\triangle P_a(t_1)A_iP_a'$. Thus, the two angles $\angle
A_iP_a(t_1)P'_a=\angle A_jP_b(t_1)P'_b$. Further, it is easy to see
$\angle oP_a(t_1)P'_a=\angle oP_b(t_1)P'_b$. Consequently, we have $\angle
oP_a(t_1)A_i=\angle oP_b(t_1)A_j$.

But, if $\angle oP_a(t_1)A_i=\angle oP_b(t_1)A_j$, then we can show that the
two functions $A_i(P_a(t))$ and $A_j(P_b(t))$ define exactly the
same curve segment. The proof is nothing but the inverse of the
above argument. Specifically, consider any time moment $t_3>t_1$
before the end of the rotation. Suppose at the moment $t_3$, $P_a$ is
at the position $P_a''$ and $P_b$ is at the position $P_b''$. Since
$\angle oP_a(t_1)A_i=\angle oP_b(t_1)A_j$ and $\angle oP_a(t_1)P''_a=\angle
oP_b(t_1)P''_b$, we have $\angle A_iP_a(t_1)P''_a=\angle A_jP_b(t_1)P''_b$. Further,
since $|A_iP_a(t_1)|=|A_jP_b(t_1)|$ and $|P_a(t_1)P_a''|=|P_b(t_1)P_b''|$, $\triangle
P_b(t_1)A_jP_b''$ is congruent to $\triangle P_a(t_1)A_iP_a''$. Thus,
$|A_iP''_a|=|A_jP''_b|$, i.e., $A_i(P_a(t))=A_j(P_b(t))$ at any time $t=t_3> t_1$.
Similarly, we can also show that at any time moment $t_3<t_1$,
$A_i(P_a(t_3))=A_j(P_b(t_3))$. Hence, $A_i(P_a(t))$ and
$A_j(P_b(t))$ define exactly the same curve segment. But this
contradicts with the fact that the curve segments defined by these two
functions are not the same. This implies that
$A_i(P_a(t))=A_j(P_b(t))$ cannot happen again at any moment $t>t_1$.

Further, without loss of generality, we assume
$A_i(P_a(t))<A_j(P_b(t))$ for any time $t<t_1$ (as
the example shown in Fig.~\ref{fig:case10}(b)). We then show that
$A_i(P_a(t_3))>A_j(P_b(t_3))$ for any time $t_3>t_1$, which means that
the two functions cross each other at their intersection point. We
briefly discuss this.  Again, suppose at the moment $t_3$, $P_a$ is
at the position $P_a''$ and $P_b$ is at the position $P_b''$. First,
since $A_i(P_a(t))<A_j(P_b(t))$ for any time $t<t_1$, it must be
$|oA_j|>|oA_i|$ (this can be proved by similar techniques as above and
we omit the details). Consider the two triangles $\triangle oA_iP_a(t_1)$ and
$\triangle oA_jP_b(t_1)$ (at the moment $t_1$). Since
$|A_iP_a(t_1)|=|A_jP_b(t_1)|$, $|oP_a(t_1)|=|oP_b(t_1)|$, and $|oA_j|>|oA_i|$, we have
$\angle oP_b(t_1)A_j>\angle oP_a(t_1)A_i$, which further implies $\angle
A_jP_b(t_1)P_b''<A_iP_a(t_1)P_a''$. Consider the triangles $\triangle
P_b(t_1)A_jP_b''$ and $\triangle P_a(t_1)A_iP_a''$. Due to $|A_iP_a(t_1)|=|A_jP_b(t_1)|$,
$|P_a(t_1)P_a''|=|P_b(t_1)P_b''|$, and $\angle A_jP_b(t_1)P_b''<A_iP_a(t_1)P_a''$,
it must be $|A_jP_b''|<|A_iP_a''|$. In
other words, $A_i(P_a(t))>A_j(P_b(t))$ at any time $t=t_3>t_1$.
\end{itemize}

The lemma thus follows.
\end{proof}

We further extend every curve segment in $F'$ into an $x$-monotone
curve, as follows. For each increasing (resp., decreasing) curve
segment, we extend it by attaching two half-lines with slope $1$
(resp., $-1$) at the two endpoints of that curve segment,
respectively, such that the resulting new curve is still monotonically increasing
(resp., decreasing). Denote the resulting new curve set by $F$. Obviously, an
increasing curve and a decreasing curve in $F$ intersect once and
they cross each other at their intersection point. For any two different increasing (resp.,
decreasing) curves in $F$, by Lemma \ref{lem:30} and the way we
extend the corresponding curve segments, they can intersect in at
most one point and cross each other at their intersection point
(if any). In other words, $F$ can be viewed as a set of pseudolines.
Let $\calA_F$ be the arrangement of $F$. Observe that the
elements in $D'_2$ are the $y$-coordinates of a subset of the
vertices of $\calA_F$. Since $\lambda_C\in D_2'$, $\lambda_C$ is the
$y$-coordinate of a vertex of $\calA_F$. Denote by $|\calA_F|$ the
number of vertices in $\calA_F$. Of course, we do not want to compute
the vertices of $\calA_F$ explicitly.  By generalizing some techniques in
\cite{ref:ColeAn89}, we have the following lemma.

\begin{Lem}\label{lem:new50}
The value $|\calA_F|$ can be computed in $O(n\log n)$ time. Given an
integer $k$ with $1\leq k\leq |\calA_F|$, the $k$-th
highest vertex of $\calA_F$ can be found in $O(n\log^2 n)$ time.
\end{Lem}
\begin{proof}
First of all, because every function in $F'$ is of $O(1)$ complexity,
we can determine in $O(1)$ time whether a curve segment in $F'$
intersects a given horizontal line, and if ``yes", then compute the
intersection. Thus, for every curve in $F$, we can also compute its
intersection with any horizontal line in $O(1)$ time. Let
$N=|F|\leq 2n$.

Recall that the curve segments in $F'$ are all in the horizontal strip between
$y=d_k$ and $y=d_{k+1}$. Thus, all vertices of $\calA_F$ above the
horizontal line $y=d_{k+1}$ are intersections of the newly attached half-lines.
We can easily determine the highest vertex of $\calA_F$ in
$O(n\log n)$ time, e.g., by using the approach in \cite{ref:ColeAn89}. Let
$l$ be a horizontal line higher than the highest vertex. Denote by
$f_1,f_2,\ldots,f_N$ the sequence of the curves of $F$ sorted in
increasing order of the $x$-coordinates of their intersections with
$l$. Similarly, we can determine the lowest vertex of $\calA_F$ in
$O(n\log n)$ time. Let $f_{\pi(1)},f_{\pi(2)},\ldots,f_{\pi(N)}$ be
the sequence of the curves of $F$ sorted in increasing order of the
$x$-coordinates of their intersections with a horizontal line below
the lowest vertex of $\calA_F$. Since the curves in $F$ can be
viewed as a set of pseudolines, as in \cite{ref:ColeAn89}, the
number of inversions in the permutation $\pi$, which can be computed
in $O(n\log n)$ time, is equal to $|\calA_F|$. In summary, we can
compute $|\calA_F|$ in $O(n\log n)$ time.

To compute the $k$-th highest vertex of $\calA_F$, we choose to
generalize the $O(n\log^2 n)$ time algorithm in \cite{ref:ColeAn89}.
Let $L$ be a set of $n$ lines in the plane and $\calA_L$ be the arrangement of $L$.
An $O(n\log^2 n)$ time algorithm was given in \cite{ref:ColeAn89} for
computing the $k$-th highest vertex of $\calA_L$ ($1\leq k\leq |\calA_L|$)
in $O(n\log^2 n)$ time based on parametric search \cite{ref:ColeSl87,ref:MegiddoAp83}. The main
property used in the algorithm \cite{ref:ColeAn89} is the following one. Denote by
$l_1,l_2,\ldots,l_n$ the sequence of lines in $L$ sorted in
increasing order of their intersections with a horizontal line above
the highest vertex of $\calA_L$. Given any horizontal line $l'$, let
$l_{\pi(1)},l_{\pi(2)},\ldots,l_{\pi(n)}$ be the sequence of lines
of $L$ sorted in increasing order of their intersections with $l'$.
Then, the number of vertices of $\calA_L$ above $l'$ is equal to the
number of inversions in the permutation $\pi$.

In our problem, since
any two curves in $F$ can intersect each other in at most one point and they
cross each other at their intersection point, the above property still holds for
$\calA_F$. Thus, the $O(n\log^2 n)$ time algorithm in
\cite{ref:ColeAn89} is applicable to our problem. Therefore, we can
find the $k$-th highest vertex of $\calA_F$ in $O(n\log^2 n)$ time, and
the lemma follows.
\end{proof}

A remark: An optimal $O(n\log n)$ time algorithm was also given in
\cite{ref:ColeAn89} (and in \cite{ref:KatzOp93}) for finding the
$k$-th highest vertex of $\calA_L$. However, these algorithms are overly
complicated. Although we think that the $O(n\log n)$ time approach in \cite{ref:ColeAn89} may
be made work for our problem, it does not
benefit our overall solution for the optimization version of the
min-sum problem because its total time is dominated by
other parts of the algorithm. Hence, the much simpler $O(n\log^2 n)$
time solution (for finding the $k$-th highest vertex of $\calA_F$)
suffices for our purpose.

Recall that $\lambda_C$ is the $y$-coordinate of a vertex of
$\calA_F$. Our algorithm for computing $\lambda_C$ then works as follows.
First, compute $|\calA_F|$. Next, find the $(|\calA_F|/2)$-th
highest vertex of $\calA_F$, and denote its $y$-coordinate by
$\lambda_m$. Determine whether $\lambda_C\leq \lambda_m$ by the
algorithm in Theorem \ref{theo:20}, after which one half of the
vertices of $\calA_F$ can be pruned away. We apply the above
procedure recursively on the remaining vertices of $\calA_F$, until
$\lambda_C$ is found. Since there are $O(\log n)$ recursive calls to this procedure,
each of which takes $O(n\log^2 n)$, the
total time for computing $\lambda_C$ is $O(n\log^3 n)$.

%Note that after having $\lambda_C$,
%the optimal solution can be obtained by applying $\lambda_C$ on the
%algorithm in Theorem \ref{theo:20}.

\begin{Theo}
The min-max optimization problem is solvable in $O(n\log^3n)$ time.
\end{Theo}

\section{The Min-sum Problem}
\label{sec:minsum}

In this section, we present our new algorithms for the min-sum problem. We show that the
boundary case of this problem is solvable in $O(n^2)$ time, which improves the
$O(n^4)$ time result in \cite{ref:TanNe10}. We also give an $O(n^2)$ time
approximation algorithm with approximation ratio $3$,
which improves the $(1+\pi)$-approximation $O(n^2)$ time
algorithm in \cite{ref:BhattacharyaOp09}.

For the boundary case, the $O(n^4)$ time algorithm in
\cite{ref:TanNe10} uses the $O(n^3)$ time Hungarian algorithm
to compute a minimum weight perfect matching in a complete bipartite graph.
However, the graph for this case is very special in the sense that all its
vertices lie on the boundary of a circle. By using the result in \cite{ref:BussLi98},
we can actually find a minimum weight perfect matching in such a graph
in $O(n)$ time. Therefore, if we follow the algorithmic scheme in
\cite{ref:TanNe10} but replace the Hungarian algorithm by the algorithm in
\cite{ref:BussLi98}, the boundary case can be solved in
$O(n^2)$ time. For completeness, more details on the proof of the
following theorem are given in Appendix \ref{app:Theorem40}.

\begin{Theo}\label{theo:40}
The boundary case of the min-sum problem can be solved in $O(n^2)$
time.
\end{Theo}

Next, we discuss our approximation algorithm for the general min-sum
problem.

Let $A_0,\ldots,A_{n-1}$ be the sensors in $C$. Our approximation
algorithm works as follows. (1) For each sensor $A_i$,
$i=0,\ldots,n-1$, compute the point $X_i$ on $\partial C$ that is
closest to $A_i$. (2) By using the algorithm in Theorem \ref{theo:40},
solve the following min-sum boundary case problem: Viewing the $n$
points $X_0,X_1,\ldots,X_{n-1}$ as {\em pseudo-sensors} (which all
lie on $\partial C$), find $n$ points on $\partial C$ as the goal
positions for the pseudo-sensors such that the sum of the distances
traveled by all $n$ pseudo-sensors is minimized. Let $X'_i$ be the goal
position for each $X_i$ ($0\leq i \leq n - 1$) in the optimal
solution thus obtained. We then let $X_i'$ be the goal position for each sensor
$A_i$, $0\leq i\leq n-1$, for our original min-sum problem. This completes the
description of our approximation algorithm.

Clearly, with Theorem \ref{theo:40}, the time complexity of the
above approximation algorithm is $O(n^2)$. The lemma below shows that the
approximation ratio of this algorithm is $3$.

\begin{Lem}\label{lem:40}
The approximation ratio of our approximation algorithm is $3$.
\end{Lem}
\begin{proof}
Let $\Delta = \sum_{i=0}^{n-1} |A_{i} X'_{i}|$. Let
$A^{\ast}_0,A^{\ast}_1,\ldots, A^{\ast}_{n-1}$ be the goal positions
of all sensors (i.e., $A^*_i$ is the goal position for each sensor
$A_i$, $0\leq i\leq n-1$) in an optimal solution for the
min-sum problem. Let $\Delta_C =\sum_{i=0}^{n-1} |A_{i} A^*_{i}|$.
Our task is to prove $\Delta\leq 3\cdot \Delta_C$.

First, \( \sum_{i=0}^{n-1} |X_{i} X'_{i}| \) $\leq$ \(
\sum_{i=0}^{n-1} |X_{i} A^{\ast}_{i}| \), and $|A_i X_i| \leq |A_i
A^{\ast}_i|$ holds for each $0 \leq i \leq n-1$. Then,
\begin{eqnarray*}
\Delta  & =    & \sum_{i=0}^{n-1} |A_{i} X'_{i}| \leq \sum_{i=0}^{n-1} (|A_{i} X_{i}| + |X_i X'_{i}|) \hspace{0.5cm} (triangle \; inequality) \\
        & =    & \sum_{i=0}^{n-1} |A_{i} X_{i}| + \sum_{i=0}^{n-1} |X_{i} X'_{i}| \leq \sum_{i=0}^{n-1} |A_{i} X_{i}| + \sum_{i=0}^{n-1} |X_{i} A^{\ast}_{i}| \\
        & \leq & 2 \cdot \sum_{i=0}^{n-1} |A_{i} X_{i}| + \sum_{i=0}^{n-1} |A_{i} A^{\ast}_{i}| \hspace{0.5cm} (triangle \; inequality) \\
        & \leq & 3 \cdot \sum_{i=0}^{n-1} |A_{i} A^{\ast}_{i}| =  3 \cdot \Delta_C.
\end{eqnarray*}
The lemma thus follows.
\end{proof}

Hence, we conclude with the following result.

\begin{Theo}
There exists an $O(n^2)$ time approximation algorithm for the
min-sum problem with approximation ratio 3.
\end{Theo}

\footnotesize
\baselineskip=11.0pt
\bibliographystyle{plain}
\bibliography{reference}

%add appendix below
\newpage
\normalsize
\appendix
\section*{Appendix}

\section{The Proof of Lemma \ref{lem:10}}
\label{app:lemmaproofs}

\noindent
{\bf Lemma \ref{lem:10}\ }\cite{ref:TanNe10}
{\em
Suppose an optimal solution for the min-max optimization problem is achieved
with $\lambda_C = |A_{i} A'_{i}|$ for some $i\in \{0,\ldots,
n-1\}$. Then either $A'_{i}$ is the point $X_{i}$, or there is
another sensor $A_{j}$ ($j \neq i$) such that $\lambda_C = |A_{j}
A'_{j}|$ also holds. In the latter case, any slight rotation of the
regular $n$-gon that achieves $\lambda_C$ in either direction
causes the value of $\lambda_C$ to increase (i.e., it makes one of the two
distances $|A_i A'_i|$ and $|A_j A'_j|$ increase and the other one decrease).
%monotonically increases one of the two distances $|A_i A'_i|$ and $|A_j A'_j|$,
%and decreases the other.
}
\vspace{0.15in}

\begin{proof}
First assume that in an optimal solution,
the sensor $A_{i}$ is the only one satisfying $\lambda_C = |A_{i} A'_{i}|$,
but $A'_i$ is not the point $X_{i}$. Thus $|A_{i} X_{i}| < |A_{i}
A'_{i}|$ holds (see Figure 1(a)).
Then, we rotate the regular $n$-gon that achieves $\lambda_C$ by moving
the vertex $A'_{i}$ towards $X_{i}$, with a very small distance $\epsilon$. Clearly,
the distance function between $A_{i}$ and $A'_{i}$ decreases monotonically during
this rotation of that $n$-gon. Denote by $A''_{0}$, $A''_{1}$, $\ldots$, $A''_{n-1}$
the new positions of the sensors after the rotation stops. Since $\epsilon$ is
arbitrarily small and $A_{i}$ is the only sensor satisfying $\lambda_C = |A_{i} A'_{i}|$,
we have $|A_{i} A''_{i}| \geq |A_{k} A''_{k}|$ for all $k \neq i$;
moreover, $|A_{i} A''_{i}| < |A_{i} A'_{i}|$ holds. But, this contradicts with
the assumption that
$\lambda_C = |A_{i} A'_{i}|$ gives an optimal solution to the min-max optimization problem.

Suppose now there exists another sensor $A_j$ such that the optimal value
$\lambda_C = |A_{j} A'_{j}|$ ($j \neq i$) also holds (see Figure 1(b)).
A slight rotation of the regular $n$-gon that achieves $\lambda_C$
in either direction cannot make $|A_{i} A''_{i}| < |A_{i} A'_{i}|$
and $|A_{j} A''_{j}| < |A_{j} A'_{j}|$ both occur, where $A''_i$ and $A''_j$
are the new positions of $A'_i$ and $A'_j$ after the rotation stops
(otherwise, it would contradict with the assumption that $\lambda_C = |A_{i} A'_{i}|=
|A_{j} A'_{j}|$ gives an optimal solution).
Hence, the rotation of the regular $n$-gon that achieves $\lambda_C$
increases one of the two distances
$|A_{i} A'_{i}|$ and $|A_{j} A'_{j}|$, while decreasing the other.
The proof is thus complete.
\end{proof}

\section{The Proof of Theorem \ref{theo:40}}
\label{app:Theorem40}

Recall that in the boundary case of the min-sum problem, all sensors
are on the boundary $\partial C$ of $C$. Let $A_0$, $A_1$, $\ldots$, $A_{n-1}$
denote the initial positions of the $n$ sensors on $\partial C$, and
$A'_0$, $A'_1$, $\ldots$, $A'_{n-1}$ denote their goal positions on
$\partial C$ that form a regular $n$-gon. Denote by $\Delta_C$ the
sum of the distances traveled by all $n$ sensors in an optimal solution
of the min-sum problem, i.e., $\Delta_C =\min \sum_{i=0}^{n-1}
|A_{i} A'_{i}|$. The following lemma has been proved in
\cite{ref:TanNe10}.

\begin{Lem}\cite{ref:TanNe10}
\label{lem-not-move}
There exists an optimal solution for the boundary case of the
min-sum problem with the following property: There exists a sensor $A_i$ which
does not move, i.e., $A_i=A_i'$.
\end{Lem}

Based on Lemma \ref{lem-not-move}, the boundary case can be solved as
follows. For each sensor $A_i$, $0\leq i\leq n-1$, let $P(A_i)$ be
the regular $n$-gon on $\partial C$ such that $A_i$ is one of its
vertices. Denote by $H_i$ the complete bipartite graph between the
set of all sensors and the set of all vertices of $P(A_i)$ such
that the weight of an edge connecting a sensor and a vertex of
$P(A_i)$ is defined as their Euclidean distance. We compute a
minimum weight perfect matching $M_i$ in $H_i$, for each $0\leq
i\leq n-1$, and the one that gives the minimum weight defines an
optimal solution for our original problem. Here, the weight of a
perfect matching is the sum of all edge weights of the matching.

The running time of the above algorithm is dominated by the step of
computing the minimum weight perfect matchings in the graphs $H_i$.
The algorithm in \cite{ref:TanNe10} uses the $O(n^3)$ time
Hungarian algorithm for computing such matchings in the graphs $H_i$.

%Haitao--->
Let $H$ be a complete bipartite graph with two vertex sets of
cardinalities $n_1$ and $n_2$, respectively, such that all its vertices
lie on the boundary of a circle and each edge weight is the
Euclidean distance between two such vertices (the edges are
represented implicitly). A maximum cardinality matching of $H$ consists of
$\min\{n_1,n_2\}$ edges. An algorithm was given in
\cite{ref:BussLi98} for computing a
minimum weight maximum cardinality matching in $H$ in $O(n_1+n_2)$ time (i.e., the total sum
of edge weights in the output maximum cardinality matching is as small as possible).

Since in our algorithm, all vertices of every complete bipartite
graph $H_i$ lie on $\partial C$, the linear time algorithm in
\cite{ref:BussLi98} can be applied to compute a minimum weight
maximum cardinality matching of $H_i$ in $O(n)$ time, for
$0\leq i\leq n-1$. (Note that a maximum
cardinality matching in the graph $H_i$ is a perfect matching, and vice versa.)
%<---
Consequently, the total running time of our algorithm is $O(n^2)$.
Theorem \ref{theo:40} thus follows.
\end{document}